\documentclass[12pt, conference, onecolumn]{IEEEtran}

\usepackage{amsmath}
\usepackage{amsthm}
\usepackage{amssymb}
\usepackage{multirow}
\usepackage{amsfonts}
\usepackage{graphicx}
\usepackage{algorithm}
\usepackage{algorithmic}
\usepackage{url}
\usepackage{lipsum}
\usepackage{subcaption}

\usepackage{latexsym}
\usepackage{color}
\usepackage{graphicx}
\usepackage{threeparttable}
\usepackage{float}
\DeclareGraphicsExtensions{.pdf,.png,.jpg}

\newcommand{\poly}{\mathrm{poly}}
\newcommand{\cG}{\mathcal{G}}
\newcommand{\cT}{\mathcal{T}}
\newcommand{\cM}{\mathcal{M}}
\newcommand{\diag}{\mathrm{diag}}
\newcommand{\bX}{\mathbf{x}}
\newcommand{\bY}{\mathbf{y}}

\newtheorem{theorem}{Theorem}
\newtheorem{definition}{Definition}
\newtheorem{lemma}{Lemma}
\newtheorem{cor}{Corollary}

\newcommand\myeq{\mathrel{\stackrel{\makebox[0pt]{\mbox{\normalfont\tiny def}}}{=}}}

\begin{document}

\title{Efficiently Decodable Non-Adaptive \\Threshold Group Testing\footnote{This paper was presented at the 2018 IEEE International Symposium on Information Theory.}\\[.7ex] 
  {\normalfont\large 
	Thach V. Bui\IEEEauthorrefmark{1}, Minoru Kuribayashi\IEEEauthorrefmark{3}, Mahdi Cheraghchi\IEEEauthorrefmark{4}, and Isao Echizen\IEEEauthorrefmark{1}\IEEEauthorrefmark{2}}\\[-1.5ex]}

\author{\IEEEauthorblockA{\IEEEauthorrefmark{1}SOKENDAI (The \\Graduate University \\for Advanced \\Studies), Hayama, \\Kanagawa, Japan\\ bvthach@nii.ac.jp}
\and
\IEEEauthorblockA{\IEEEauthorrefmark{3}Graduate School\\ of Natural Science\\ and Technology, \\Okayama University, \\Okayama, Japan\\kminoru@okayama-u.ac.jp}
\and
\IEEEauthorblockA{\IEEEauthorrefmark{4}Department of \\Computing, Imperial \\College London, UK\\m.cheraghchi@imperial.ac.uk}
\and
\IEEEauthorblockA{\IEEEauthorrefmark{2}National Institute\\ of Informatics, \\Tokyo, Japan \\ iechizen@nii.ac.jp}}

\maketitle

\thispagestyle{plain}
\pagestyle{plain}

\begin{abstract}
We consider non-adaptive threshold group testing for identification of up to $d$ defective items in a set of $n$ items, where a test is positive if it contains at least $2 \leq u \leq d$ defective items, and negative otherwise. The defective items can be identified using $t = O \left( \left( \frac{d}{u} \right)^u \left( \frac{d}{d - u} \right)^{d-u} \left(u \log{\frac{d}{u}} + \log{\frac{1}{\epsilon}} \right) \cdot d^2 \log{n} \right)$ tests with probability at least $1 - \epsilon$ for any $\epsilon > 0$ or $t = O \left(  \left( \frac{d}{u} \right)^u \left( \frac{d}{d -u} \right)^{d - u} d^3 \log{n} \cdot \log{\frac{n}{d}} \right)$ tests with probability 1. The decoding time is $t \times \mathrm{poly}(d^2 \log{n})$. This result significantly improves the best known results for decoding non-adaptive threshold group testing: $O(n\log{n} + n \log{\frac{1}{\epsilon}})$ for probabilistic decoding, where $\epsilon > 0$, and $O(n^u \log{n})$ for deterministic decoding.
\end{abstract}

\section{Introduction}
\label{sec:intro}

The goal of combinatorial group testing is to identify up to $d$ defective items among a population of $n$ items (usually $d$ is much smaller than $n$). This problem dates back to the work of Dorfman~\cite{dorfman1943detection}, who proposed using a pooling strategy to identify defectives in a collection of blood samples. In each test, a group of items are pooled, and the combination is tested. The result is positive if at least one item in the group is defective and is otherwise negative. Damaschke~\cite{damaschke2006threshold} introduced a generalization of classical group testing known as \textit{threshold group testing}. In this variation, the result is positive if the corresponding group contains at least $u$ defective items, where $u$ is a parameter, is negative if the group contains no more than $\ell$ defective items, where $0 \leq \ell < u$, and is arbitrary otherwise. When $u = 1$ and $\ell = 0$, threshold group testing reduces to classical group testing. We note that $\ell$ is always smaller than the number of defective items. Otherwise, every test would yield a negative outcome and no information can be extracted from the test outcomes.

There are two approaches for the design of tests. The first is \textit{adaptive group testing} in which there are several testing stages, and the design of each stage depends on the outcomes of the previous stages. The second is \textit{non-adaptive group testing} (NAGT) in which all tests are designed in advance, and the tests are performed in parallel. NAGT is appealing to researchers in most application areas, such as computational and molecular biology~\cite{farach1997group}, multiple access communications \cite{wolf1985born} and data streaming~\cite{cormode2005s} (cf.\ \cite{du2000combinatorial}). The focus of this work is on NAGT.

In both threshold and classical group testing, it is desirable to minimize the number of tests and, to efficiently identify the set of defective items (i.e., have an efficient decoding algorithm). For both testings, one needs $\Omega(d\log{n})$ tests to identify all defective items~\cite{du2000combinatorial,chen2011almost,chang2011reconstruction} using adaptive schemes. In adaptive schemes, the decoding algorithm is usually implicit in the test design. The number of tests and the decoding time are significantly different between classical non-adaptive (CNAGT) and non-adaptive threshold group testing (NATGT).

In CNAGT, Porat and Rothschild~\cite{porat2008explicit} first proposed explicit nonadaptive constructions using $O(d^2 \log{n})$ tests. However, there is no efficient (sublinear-time) decoding algorithm associated with their schemes. For exact identification, there are explicit schemes allowing defective items be identified using $\poly(d, \log{n})$ tests in time $\poly(d, \log{n})$~\cite{ngo2011efficiently,cheraghchi2013noise} (the number of tests can be as low as $O(d^{1 + o(1)} \log{n})$ if false positives are allowed in the reconstruction). To achieve a nearly optimal number of tests in adaptive group testing and with low decoding complexity, Cai et al.~\cite{cai2013grotesque} proposed using probabilistic schemes that need $O(d \log{d} \cdot \log{n})$ tests to find the defective items in time $O(d(\log{n} + \log^2{d}))$.

In threshold group testing, Damaschke~\cite{damaschke2006threshold} showed that the set of positive items can be identified with $\binom{n}{u}$ tests with up to $g$ false positives and $g$ false negatives, where $g = u - \ell - 1$ is the \textit{gap} parameter. Cheraghchi~\cite{cheraghchi2013improved} showed that it is possible to find the defective items with $O(d^{g+2} \log{d} \cdot \log(n/d))$ tests, and that this trade-off is essentially optimal. Recently, De Marco et al.~\cite{de2017subquadratic} improved this bound to $O(d^{3/2} \log(n/d))$ tests under the extra assumption that the number of defective items is exactly $d$, which is rather restrictive in application. Although the number of tests has been extensively studied, there have been few reports that focus on the decoding algorithm as well. Chen and Fu~\cite{chen2009nonadaptive} proposed schemes based on CNAGT for when $g = 0$ that can find the defective items using $O \left( \left( \frac{d}{u} \right)^u \left( \frac{d}{d-u} \right)^{d - u} d \log{\frac{n}{d}} \right)$ tests in time $O(n^u \log{n}) $. Chan et al.~\cite{chan2013stochastic} presented a randomized algorithm with $O\left( \log{\frac{1}{\epsilon}} \cdot d\sqrt{u} \log{n}\right)$ tests to find the defective items in time $O(n\log{n} + n \log{\frac{1}{\epsilon}})$ given that the number of defective items is exactly $d$, $g = 0$, and $u = o(d)$. These conditions are too strict to apply in practice. Moreover, the cost of these decoding schemes increases with $n$. Our objective is to find an efficient decoding scheme to identify up to $d$ defective items in NATGT when $g = 0$.

\textit{\textbf{Contributions:}} In this paper, we consider the case where $g = 0$, i.e., $\ell = u-1$ ($u \geq 2$), and call this model $u$-NATGT.  We first propose an efficient scheme for identifying up to $d$ defective items in NATGT in time $t \times \poly(d^2  \log{n})$, where $t$ is the number of tests. Our main idea is to create at least a specified number of rows in the test matrix such that the corresponding test in each row contains exactly $u$ defective items and such that the defective items in the rows are the defective items to be identified. We ``map'' these rows using a special matrix constructed from a disjunct matrix (defined later) and its complementary matrix, thereby converting the outcome in NATGT to the outcome in CNAGT. The defective items in each row can then be efficiently identified.

Although Cheraghchi~\cite{cheraghchi2013improved}, De Marco et al.~\cite{de2017subquadratic}, and D'yachkov et al.~\cite{d2013superimposed} proposed nearly optimal bounds on the number of tests, there are no decoding algorithms associated with their schemes. Note that the number of tests is optimal in~\cite{d2013superimposed}, i.e., $O(d^2 \log{n})$, when the threshold $u$ is a fixed constant. On the other hand, the scheme of Chen et al.~\cite{chen2009nonadaptive} requires a smaller number of tests compared with our scheme. However, the decoding complexity of their scheme is exponential in the number of items $n$, which is impractical. Chan et al.~\cite{chan2013stochastic} proposed a probabilistic approach to achieve a small number of tests, which combinatorially can be better than our scheme. However, their scheme is only applicable when the number of defective items is exactly $d$, the threshold $u$ is much smaller than $d$ ($u = o(d)$), and the decoding complexity remains high, namely $O(n \log{n} + n \log{\frac{1}{\epsilon}})$, where $\epsilon > 0$ is the precision parameter.

We present a divide and conquer scheme based on the main idea which we then instantiate via deterministic and randomized decoding. Deterministic decoding is a deterministic scheme in which all defective items can be found with probability 1. Randomized decoding reduces the number of tests; all defective items can be found with probability at least $1 - \epsilon$ for any $\epsilon > 0$. The decoding complexity is $t \times \poly(d^2 \log{n})$. A comparison with existing work is given in Table~\ref{tbl:comparison}.

\begin{table*}[ht]
\caption{Comparison with existing work.}

\begin{center}
\scalebox{1.0}{
\begin{tabular}{|l|c|c|c|c|}
\hline
Scheme & \begin{tabular}{@{}c@{}} Number of \\defective items \end{tabular} & \begin{tabular}{@{}c@{}} Number of tests \\ $t$ \end{tabular}  & Decoding complexity & Decoding type \\
\hline
Cheraghchi~\cite{cheraghchi2013improved} & $\leq d$ & $O(d^2 \log{d} \cdot \log{\frac{n}{d}})$ & $\times$ & $\times$ \\
\hline
De Marco et al.~\cite{de2017subquadratic} & $d$ & $O \left( d^2 \cdot \sqrt{\frac{d-u}{du}} \cdot \log{\frac{n}{d}} \right)$ & $\times$ & $\times$ \\
\hline
D'yachkov et al.~\cite{d2013superimposed} & $\leq d$ & $O \left( d^2 \log{n} \cdot \frac{(u-1)! 4^u}{(u - 2)^u (\ln{2})^u}\right)$ & $\times$ & $\times$ \\
\hline
Chen et al.~\cite{chen2009nonadaptive} & $\leq d$ & $O \left( \left( \frac{d}{u} \right)^u \left( \frac{d}{d-u} \right)^{d - u} d \log{\frac{n}{d}} \right)$ & $O(n^u \log{n}) $ & Deterministic \\
\hline
\textbf{Deterministic decoding}& $\leq d$ & $O \left(  \left( \frac{d}{u} \right)^u \left( \frac{d}{d -u} \right)^{d - u} d^3 \log{n} \cdot \log{\frac{n}{d}} \right)$ & $t \times \poly(d^2 \log{n})$ & Deterministic \\
\hline
Chan et al.~\cite{chan2013stochastic} & $d$ & $O\left( \log{\frac{1}{\epsilon}} \cdot d\sqrt{u} \log{n}\right)$ & $O(n\log{n} + n \log{\frac{1}{\epsilon}})$ & Random \\
\hline
\textbf{Randomized decoding} & $\leq d$ & $O \left( \left( \frac{d}{u} \right)^u \left( \frac{d}{d - u} \right)^{d-u} \left(u \log{\frac{d}{u}} + \log{\frac{1}{\epsilon}} \right) \cdot d^2 \log{n} \right)$ & $t \times \poly(d^2 \log{n})$ & Random \\
\hline
\end{tabular}}

\label{tbl:comparison}
\end{center}
\end{table*}

\section{Preliminaries}
\label{sec:pre}

For consistency, we use capital calligraphic letters for matrices, non-capital letters for scalars, bold letters for vectors, and capital letters for sets. All matrix and vector entries are binary. Here are some of the notations used:
\begin{enumerate}
\item $n, d, \bX = (x_1, \ldots, x_n)^T$: number of items, maximum number of defective items, and binary representation of $n$ items.
\item $S = \{j_1, j_2, \ldots, j_{|S|} \}$: the set of defective items; cardinality of $S$ is $|S| \leq d$.
\item $\otimes, \odot$: operation related to $u$-NATGT and CNAGT, to be defined later.
\item $\cT$: $t \times n$ measurement matrix used to identify up to $d$ defective items in $u$-NATGT, where integer $t \geq 1$ is the number of tests.
\item $\cG = (g_{ij})$: $h \times n$ matrix, where $h \geq 1$.
\item $\cM = (m_{ij})$: $k \times n$ $(d+1)$-disjunct matrix used to identify up to $u$ defective items in $u$-NATGT and $(d+1)$ defective items in CNAGT, where integer $k \geq 1$ is the number of tests.
\item $\overline{\cM} = (\overline{m}_{ij})$: the $k \times n$ complementary matrix of $\cM$; $\overline{m}_{ij} = 1 - m_{ij}$.
\item $\cT_{i, *}, \cG_{i, *}, \cM_{i,*}, \cM_j$: row $i$ of matrix $\cT$, row $i$ of matrix $\cG$, row $i$ of matrix $\cM$, and column $j$ of matrix $\cM$, respectively.
\item $\bX_i = (x_{i1}, \ldots, x_{in})^T, S_i$: binary representation of items and set of indices of defective items in row $\cG_{i, *}$. 
\item $\diag(\cG_{i, *}) = \diag(g_{i1}, \ldots, g_{in})$: diagonal matrix constructed by input vector $\cG_{i, *}$.
\end{enumerate}

\subsection{Problem definition}
\label{sub:probDef}
We index the population of $n$ items from 1 to $n$. Let $[n] = \{1, 2, \ldots, n \}$ and $S$ be the defective set, where $|S| \leq d$. A test is defined by a subset of items $P \subseteq [n]$. $(d, u, n)$-NATGT is a problem in which there are up to $d$ defective items among $n$ items. A test consisting of a subset of $n$ items is positive if there are at least $u$ defective items in the test, and each test is designed in advance. Formally, the test outcome is positive if $|P \cap S| \geq u$ and negative if $|P \cap S| < u$. 

We can model $(d, u, n)$-NATGT as follows: A $t \times n$ binary matrix $\cT =(t_{ij})$ is defined as a measurement matrix, where $n$ is the number of items and $t$ is the number of tests. Vector $\bX = (x_1,\ldots,x_n)^T$ is the binary representation vector of $n$ items, where $|\bX| \leq d$. An entry $x_j=1$ indicates that item $j$ is defective, and $x_j=0$ indicates otherwise. The $j$th item corresponds to the $j$th column of the matrix. An entry $t_{ij}=1$ naturally means that item $j$ belongs to test $i$, and $t_{ij}=0$ means otherwise. The outcome of all tests is $\bY=(y_1, \ldots, y_t)^T$, where $y_i=1$ if test $i$ is positive and $y_i=0$ otherwise. The procedure to get the outcome vector $\bY$ is called the \textit{encoding procedure}. The procedure used to identify defective items from $\bY$ is called the \textit{decoding procedure.} Outcome vector $\bY$ is

\begin{equation}
\label{eqn:thresholdGT}
\bY = \cT \otimes \bX \myeq \begin{bmatrix}
\cT_{1, *} \otimes \bX \\
\vdots \\
\cT_{t, *} \otimes \bX
\end{bmatrix} \myeq \begin{bmatrix}
y_1 \\
\vdots \\
y_t
\end{bmatrix}
\end{equation}
where $\otimes$ is a notation for the test operation in $u$-NATGT; namely, $y_i = \cT_{i, *} \otimes \bX = 1$ if $\sum_{j=1}^n x_j t_{ij} \geq u$, and $y_i = \cT_{i, *} \otimes \bX = 0$ if $\sum_{j=1}^n x_j t_{ij} < u$ for $i=1, \ldots, t$. Our objective is to find an efficient decoding scheme to identify up to $d$ defective items in $(d, u, n)$-NATGT.

\subsection{Disjunct matrices}
\label{sub:disjunct}

When $u = 1$, $u$-NATGT reduces to CNAGT. To distinguish CNAGT and $u$-NATGT, we change notation $\otimes$ to $\odot$ and use a $k \times n$ measurement matrix $\cM$ instead of the $t \times n$ matrix $\cT$. The outcome vector $\bY$ in~\eqref{eqn:thresholdGT} is equal to
\begin{equation}
\label{eqn:disjunct}
\bY = \cM \odot \bX \myeq \begin{bmatrix}
\cM_{1, *} \odot \bX \\
\vdots \\
\cM_{k, *} \odot \bX
\end{bmatrix}
\myeq \begin{bmatrix}
\bigvee_{j=1}^{n} x_j \wedge m_{1j} \\
\vdots \\
\bigvee_{j=1}^{n} x_j \wedge m_{kj}
\end{bmatrix} = \bigvee_{j=1, x_j = 1}^{n} \cM_j \myeq \begin{bmatrix}
y_1 \\
\vdots \\
y_k
\end{bmatrix}
\end{equation}
where $\odot$ is the Boolean operator for vector multiplication in which multiplication is replaced with the AND ($\wedge$) operator and addition is replaced with the OR ($\vee$) operator, and $y_i = \cM_{i, *} \odot \bX = \bigvee_{j=1}^{n} x_j \wedge m_{ij} = \bigvee_{j=1, x_j = 1}^{n} m_{ij}$ for $i = 1, \ldots, k$. 

The union of $r$ columns of $\cM$ is defined as follows: $\bigvee_{i=1}^r \cM_{j_i} = \left( \bigvee_{i=1}^r m_{1 j_i}, \ldots, \bigvee_{i=1}^r m_{t j_i} \right)^T$. A column is said to not be included in another column if there exists a row such that the entry in the first column is 1 and the entry in the second column is 0. If $\cM$ is an $(d+1)$-disjunct matrix satisfying the property that the union of up to $(d+1)$ columns does not include any remaining column, $\bX$ can always be recovered from $\bY$. We need $\cM$ to be an $(d+1)$-disjunct matrix that can be efficiently decoded, as in~\cite{cheraghchi2013noise,ngo2011efficiently}, to identify up to $d$ defective items in $u$-NATGT. A $k \times n$ strongly explicit matrix is a matrix in which the entries can be computed in time $\poly(k)$. We can now state the following theorem:
\begin{theorem}~\cite[Theorem 16]{ngo2011efficiently}
\label{thr:dDisjunct}
Let $1 \leq d \leq n$. There exists a strongly explicit $k \times n$ $(d+1)$-disjunct matrix with $k = O(d^2 \log{n})$ such that for any $k \times 1$ input vector, the decoding procedure returns the set of defective items if the input vector is the union of up to $d+1$ columns of the matrix in $\poly(k)$ time. Moreover, each column of $\cM$ can be generated in time $O(k^2 \log{N})$.
\end{theorem}

\subsection{Completely separating matrix}
\label{sub:separating}
We now introduce the notion of completely separating matrices which are used to get efficient decoding algorithms for $(d, u, n)$-NATGT. An $(u, w)$-completely separating matrix is defined as follows:

\begin{definition}
\label{def:separatingMatrix}
An $h \times n$ matrix $\cG = (g_{ij})_{1 \leq i \leq h, 1 \leq j \leq n}$ is called an $(u, w)$-completely separating matrix if for any pair of subsets $I, J \subset [n]$ such that $|I| = u$, $|J| = w$, and $I \cap J = \emptyset$, there exists row $l$ such that $g_{lr} = 1$ for any $r \in I$ and $g_{ls} = 0$ for any $s \in J$. Row $l$ is called a \textit{singular} row to subsets $I$ and $J$. When $u = 1$, the matrix $\cG$ is called a $w$-disjunct matrix.
\end{definition}

This definition is slightly different from the one described by D'yachkov et al.~\cite{d2002families}. It is easy to verify that, if a matrix is an $(u, w)$-completely separating matrix, it is also an $(u, v)$-completely separating matrix for any $v \leq w$. Below we present the existence of such matrices.

\begin{theorem}
\label{thr:existingSeparatingMatrix}
Given integers $1 \leq u + w < n$, there exists an $(u, w)$-completely separating matrix of size $h \times n$, where
\begin{eqnarray}
h &=& \frac{ (u + w)^{u + w}}{u^u w^w} \left( (u + w)\log{\frac{\mathrm{e}n}{u + w}} + u\log{\frac{\mathrm{e}(u + w)}{u}} \right) + 1 \notag
\end{eqnarray}
and $\mathrm{e}$ is base of the natural logarithm.
\end{theorem}

\begin{proof}
Consider a randomly generated $h \times n$ matrix $\cG = (g_{ij})_{1 \leq i \leq h, 1 \leq j \leq n}$ in which each entry $g_{ij}$ is assigned to 1 with probability $p$ and to 0 with probability $1 - p$. For any pair of subsets $I, J \subset [n]$ such that $|I| = u$, $|J| = w$, the probability that a row is not singular is
\begin{equation}
1 - p^u (1 - p)^w.
\end{equation}
Subsequently, the probability that there is no singular row to subsets $I$ and $J$ is
\begin{equation}
f(p) = \left( 1 - p^u (1 - p)^w \right)^h.
\end{equation}

Using a union bound, the probability that any pair of subsets $I, J \subset [n]$, where $|I| = u$, $|J| = w$, does not have a singular row; i.e., the probability that $\cG$ is not an $(u, w)$-separating matrix, is
\begin{equation}
g(p, h, u, w, n) = \binom{n}{u + w} \binom{u + w}{u} f(p) = \binom{n}{u + w} \binom{u + w}{u} \left( 1 - p^u (1 - p)^w \right)^h.
\end{equation}
To ensure that there exists an $(u, w)$-separating matrix $\cG$, one needs to find $p$ and $h$ such that $g(p, h, u, w, n) < 1$. Choosing $p = \frac{u}{u + w}$, we have:
\begin{eqnarray}
f(p) &=& \left( 1 - p^u (1 - p)^w \right)^h = \left( 1 - \left( \frac{u}{u + w} \right)^u \left(1 - \frac{u}{u + w} \right)^{w} \right)^h \nonumber \\ 
&\leq& \mathrm{exp} \left(- h \cdot \frac{u^u w^w}{(u + w)^{u + w}} \right), \label{eqn:Exist2}
\end{eqnarray}
where~\eqref{eqn:Exist2} holds because $1 - x \leq \mathrm{e}^{-x}$ for any $x > 0$. Thus we have 
\begin{alignat}{3}
&& g(p, h, u, w, n) &= \binom{n}{u + w} \binom{u + w}{u} f(p) \leq \left( \frac{\mathrm{e}n}{u + w} \right)^{u + w} \left( \frac{\mathrm{e}(u + w)}{u} \right)^{u} f(p) \label{eqn:Exist4} \\
&& &\leq \left( \frac{\mathrm{e}n}{u + w} \right)^{u + w} \left( \frac{\mathrm{e}(u + w)}{u} \right)^{u} \mathrm{exp} \left(- h \cdot \frac{u^u w^w}{(u + w)^{u + w}} \right) \label{eqn:Exist5} \\
&& &< 1 \nonumber \\ 
\Longleftrightarrow&& \left( \frac{\mathrm{e}n}{u + w} \right)^{u + w} \left( \frac{\mathrm{e}(u + w)}{u} \right)^{u} &<  \mathrm{exp} \left(h \cdot \frac{u^u w^w}{(u + w)^{u + w}} \right) \nonumber \\ 
\Longleftrightarrow&& h &> \frac{ (u + w)^{u + w}}{u^u w^w} \left( (u + w)\log{\frac{\mathrm{e}n}{u + w}} + u\log{\frac{\mathrm{e}(u + w)}{u}} \right)  \label{eqn:Exist7}
\end{alignat}
In the above, we have \eqref{eqn:Exist4} because $\binom{a}{b} \leq \left( \frac{\mathrm{e} a}{b} \right)^b$ and \eqref{eqn:Exist5} by using \eqref{eqn:Exist2}. From \eqref{eqn:Exist7}, if we choose 
\begin{align}
h &= \frac{ (u + w)^{u + w}}{u^u w^w} \left( (u + w)\log{\frac{\mathrm{e}n}{u + w}} + u\log{\frac{\mathrm{e}(u + w)}{u}} \right) + 1 \\
&= O \left( \mathrm{e}^{u + w} (u + w) \log{\frac{n}{u + w}} \right), \mbox{ because } u + w < n,
\end{align}
then $g(p, h, u, w, n) < 1$; i.e., there exists an $(u, w)$-completely separating matrix of size $h \times n$.
\end{proof}

Suppose that $\cG$ is an $h \times n$ $(u, w)$-completely separating matrix. If $w$ is set to $d-u$, then \textbf{every} $h \times d$ submatrix, which is constructed by any $d$ columns, is an $(u, d- u)$-completely separating matrix. This property is too strong and increases the number of rows in $\cG$. To reduce the number of rows, we relax this property to a ``for-each'' guarantee as follows: \textbf{each} $h \times d$ submatrix, which is constructed by $d$ columns of $\cG$, is an $(u, d- u)$-completely separating matrix with high probability. The following corollary describes this idea in more detail.

\begin{cor}
\label{cor:rndMatrix}
Let $u, d, n$ be any given positive integers such that $1 \leq u < d < n$. For any $\epsilon > 0$, there exists a random $h \times n$ matrix such that for each $h \times d$ submatrix, which is constructed picking a set of $d$ columns, is an $(u, d- u)$-completely separating matrix with probability at least $1 - \epsilon$, where
\begin{eqnarray}
h &=&  \left( \frac{d}{u} \right)^u \left( \frac{d}{d - u} \right)^{d-u} \left(u \log{\frac{\mathrm{e}d}{u}} + \log{\frac{1}{\epsilon}} \right) \nonumber
\end{eqnarray}
and $\mathrm{e}$ is base of the natural logarithm.
\end{cor}

\begin{proof}
Consider a random $h \times n$ matrix $\cG = (g_{ij})_{1 \leq i \leq h, 1 \leq j \leq n}$ in which each entry $g_{ij}$ is assigned to 1 with probability of $\frac{u}{d}$ and to 0 with probability of $1 - \frac{u}{d}$. Our task is to prove that \textit{each} $h \times d$ matrix $\cG^\prime$, constructed by $d$ columns of $\cG$, is an $(u, d - u)$-completely separating matrix with probability at least $1 - \epsilon$ for any $\epsilon > 0$. Specifically, we prove that $h = \left( \frac{d}{u} \right)^u \left( \frac{d}{d - u} \right)^{d-u} \left(u \log{\frac{\mathrm{e}d}{u}} + \log{\frac{1}{\epsilon}} \right)$ is sufficient to achieve such $\cG^\prime$. Similar to the proof in Theorem~\ref{thr:existingSeparatingMatrix}, the probability that $\cG^\prime$ is not an $(u, d - u)$-completely separating matrix up to $\epsilon$ is
\begin{alignat}{3}
&& \binom{d}{u} \left(1 - \left( \frac{u}{d} \right)^u \left(1 - \frac{u}{d} \right)^{d - u} \right)^h &\leq \left( \frac{\mathrm{e}d}{u} \right)^{u} \mathrm{exp} \left( -h \left( \frac{u}{d} \right)^u \left( \frac{d - u}{d} \right)^{d - u} \right) \leq \epsilon \label{eqn:cor1} \\
\Longleftrightarrow&& \frac{1}{\epsilon} \left( \frac{\mathrm{e}d}{u} \right)^{u} &\leq  \mathrm{exp} \left(h \left( \frac{u}{d} \right)^u \left( \frac{d - u}{d} \right)^{d-u} \right) \nonumber \\
\Longleftrightarrow&& h &\geq \left( \frac{d}{u} \right)^u \left( \frac{d}{d - u} \right)^{d-u} \left(u \log{\frac{\mathrm{e}d}{u}} + \log{\frac{1}{\epsilon}} \right) \label{eqn:cor4}
\end{alignat}
We get \eqref{eqn:cor1} because $1 - x \leq e^{-x}$ for any $x > 0$ and $\binom{a}{b} \leq \left( \frac{\mathrm{e} a}{b} \right)^b$. This completes the proof.
\end{proof}

\section{Proposed scheme}
\label{sec:DC}
The basic idea of our scheme, which uses a divide and conquer strategy, is to create at least $\kappa$ rows of matrix $\cG$, e.g., $i_1, i_2, \ldots, i_\kappa$ such that $|S_{i_1}| = \cdots = |S_{i_\kappa}| = u$ and $S_{i_1} \cup \ldots \cup S_{i_\kappa} = S$. Then we ``map'' these rows by using a special matrix that enables us to convert the outcome in NATGT to the outcome in CNAGT. The defective items in each row can then be efficiently identified. We present a particular matrix that achieves efficient decoding for each row in the following section. This idea is illustrated in Fig.~\ref{fig:Proposed}.

\begin{figure}[ht]
\centering
  \includegraphics[scale=0.5]{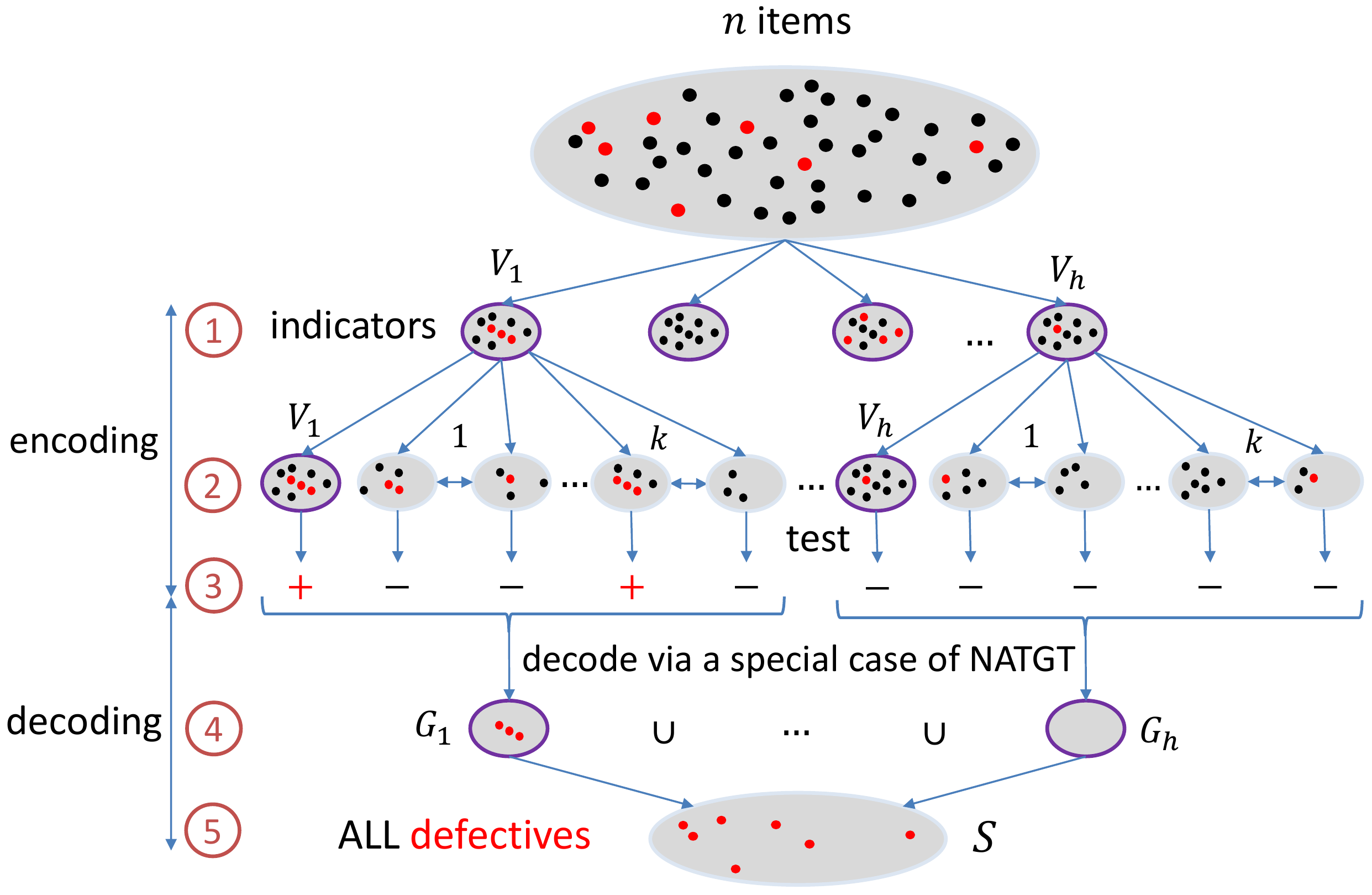}
  \caption{An illustration of the proposed scheme.}
  \label{fig:Proposed} 
\end{figure}

In Fig.~\ref{fig:Proposed}, the defective and negative items are represented as red and black dots, respectively. There are five steps in the proposed scheme. The encoding procedure includes Steps 1, 2, and 3. Steps 4 and 5 correspond to the decoding procedure. A row $\cG_{i, *}$ can be represented by a support set $V_i = \{j \mid g_{ij} = 1 \}$. Let $\cM = (m_{i^\prime j})$ be a $k \times n$ $(d + 1)$-disjunct matrix and $M_{i^\prime} = \{ j \mid m_{i^\prime j} = 1\}$ be the support set of $\cM_j$ for $i^\prime = 1, \ldots, k$ and $j = 1, \ldots, n$.

In the encoding procedure, we create $h$ ``indicating subsets'' $V_1, \ldots, V_h$ in Step 1 in which $V_{i_l} \cap S_{i_l} = S_{i_l}$ for $l = 1, \ldots, \kappa$. Our objective is to extract $S_{i_l}$ from $V_{i_l}$ efficiently. In Step 2, each subset $V_i$ is mapped to $2k + 1$ subsets. They are the $V_i$ and $k$ dual subsets in which each dual subset ($V_i \cap M_{i^\prime}$ and $V_i \setminus V_i \cap M_{i^\prime}$ for $i^\prime = 1, \ldots, k$) is a partition of a $V_i$ created from $\cM$. Step 3 simply gets the outcomes of all tests generated in Step 2.

In the decoding procedure, Step 4 gets the defective set $G_i$ from the $2k + 1$ subsets created from $V_i$ as an instance of NATGT, for $i = 1, \ldots, h$. As a result, the cardinality of $G_i$ is either $u$ or 0. Finally, the defective set $S$ is the union of $G_1, \ldots, G_h$ in Step 5.

\subsection{When the number of defective items equals the threshold}
\label{sub:specialCase}

In this section, we consider a special case in which the number of defective items equals the threshold, i.e., $|\bX| = u$. Given a measurement matrix $\cM$ and a representation vector of $u$ defective items $\bX$ ($|\bX| = u$), what we observe is $\bY \myeq \cM \otimes \bX \myeq (y_1, \ldots, y_k)^T$. Our objective is to recover $\bY^\prime \myeq \cM \odot \bX \myeq (y^\prime_1, \ldots, y^\prime_k)^T$ from $\bY$. Then $\bX$ can be recovered if we choose $\cM$ as an $(d+1)$-disjunct matrix described in Theorem~\ref{thr:dDisjunct}. To achieve this goal, we create a measurement matrix:
\begin{equation}
\label{eqn:elementaryMatrix}
\mathcal{A} \myeq \begin{bmatrix}
\cM \\
\overline{\cM}
\end{bmatrix}
\end{equation}
where $\cM = (m_{ij})$ is a $k \times n$ $(d+1)$-disjunct matrix as described in Theorem~\ref{thr:dDisjunct} and $\overline{\cM} = (\overline{m}_{ij})$ is the complement matrix of $\cM$, $\overline{m}_{ij} = 1 - m_{ij}$ for $i = 1, \ldots, k$ and $j = 1, \ldots, n$. We note that $\cM$ can be decoded in time $\poly(k) = \poly(d^2 \log{n})$ because $k = O(d^2 \log{n})$. Let us assume that the outcome vector is $\mathbf{z}$. Then we have: 
\begin{equation}
\label{eqn:special1}
\mathbf{z} \myeq \mathcal{A} \otimes \bX \myeq \begin{bmatrix}
\cM \otimes \bX \\
\overline{\cM} \otimes \bX
\end{bmatrix} \myeq \begin{bmatrix}
\bY \\
\overline{\bY}
\end{bmatrix}
\end{equation}
where $\bY = \cM \otimes \bX = (y_1, \ldots, y_k)^T$ and $\overline{\bY} = \overline{\cM} \otimes \bX = (\overline{y}_1, \ldots, \overline{y}_k)^T$. The following lemma shows that $\bY^\prime = \cM \odot \bX$ is always obtained from $\mathbf{z}$; i.e., vector $\bX$ can always be recovered.

\begin{lemma}
\label{lem:specialU}
Given integers $2 \leq u \leq d < n$, there exists a strongly explicit $2k \times n$ matrix such that if there are exactly $u$ defective items among $n$ items in $u$-NATGT, the $u$ defective items can be identified in time $\poly(k)$, where $k = O(d^2 \log{n})$.
\end{lemma}

\begin{proof}
We construct the measurement matrix $\mathcal{A}$ in \eqref{eqn:elementaryMatrix} and assume that $\mathbf{z}$ is the observed vector as in \eqref{eqn:special1}. Our task is to create vector $\bY^\prime = \cM \odot \bX$ from $\mathbf{z}$. One can get it using the following rules, where $l = 1, 2, \ldots, k$: 
\begin{enumerate}
\item If $y_l = 1$, then $y^\prime_l = 1$.
\item If $y_l = 0$ and $\overline{y}_l = 1$, then $y^\prime_l = 0$.
\item If $y_l = 0$ and $\overline{y}_l = 0$, then $y^\prime_l = 1$.
\end{enumerate}
We now prove the correctness of the above rules. Because $y_l = 1$, there are at least $u$ defective items in row $\cM_{l, *}$. Then, the first rule is implied.

If $y_l = 0$, there are less than $u$ defective items in row $\cM_{l, *}$. Because $|\bX| = u$, we have $\overline{y}_l = 1$, and the threshold is $u$, there must be $u$ defective items in row $\overline{\cM}_{l, *}$. Moreover, since $\overline{\cM}_{l, *}$ is the complement of $\cM_{l, *}$, there must be no defective item in test $l$ of $\cM$. Therefore, we have $y^\prime_l = 0$, and the second rule is implied.

If $y_l = 0$, there are less than $u$ defective items in row $\cM_{l, *}$. Similarly, if $\overline{y}_l = 0$, there are less than $u$ defective items in row $\overline{\cM}_{l, *}$. Because $\overline{\cM}_{l, *}$ is the complement of $\cM_{l, *}$, the number of defective items in row $\cM_{l, *}$ or $\overline{\cM}_{l, *}$ cannot be equal to zero, since either $y_l$ would equal $1$ or $\overline{y}_l$ would equal $1$. Since the number of defective items in row $\cM_{l, *}$ is not equal to zero, the test outcome is positive, i.e., $\bY^\prime_{l} = 1$. The third rule is thus implied.

Since we get $\bY^\prime = \cM \odot \bX$, the matrix $\cM$ is an $(d+1)$-disjunct matrix and $u \leq d$, $u$ defective items can be identified in time $\poly(k)$ by Theorem~\ref{thr:dDisjunct}.
\end{proof}

\textit{\textbf{Example:}} We demonstrate Lemma~\ref{lem:specialU} by setting $u = d = 2$, $k = 9$, and $n = 12$ and defining a $9 \times 12$ 2-disjunct matrix $\cM$ with the first two columns as follows:
\begin{equation}
\label{eqn:exSpecial}
\cM = \left[
\begin{tabular}{cccccccccccc}
0 & 0 & 0 & 0 & 0 & 0 & 1 & 1 & 1 & 1 & 0 & 0 \\
0 & 0 & 0 & 1 & 1 & 1 & 0 & 0 & 0 & 1 & 0 & 0 \\
1 & 1 & 1 & 0 & 0 & 0 & 0 & 0 & 0 & 1 & 0 & 0 \\
0 & 0 & 1 & 0 & 0 & 1 & 0 & 0 & 1 & 0 & 1 & 0 \\
0 & 1 & 0 & 0 & 1 & 0 & 0 & 1 & 0 & 0 & 1 & 0 \\
1 & 0 & 0 & 1 & 0 & 0 & 1 & 0 & 0 & 0 & 1 & 0 \\
0 & 1 & 0 & 1 & 0 & 0 & 0 & 0 & 1 & 0 & 0 & 1 \\
0 & 0 & 1 & 0 & 1 & 0 & 1 & 0 & 0 & 0 & 0 & 1 \\
1 & 0 & 0 & 0 & 0 & 1 & 0 & 1 & 0 & 0 & 0 & 1 
\end{tabular}
\right], \bY = \begin{bmatrix}
0 \\
0 \\
1 \\
0 \\
0 \\
0 \\
0 \\
0 \\
0
\end{bmatrix}, \overline{\bY} = \begin{bmatrix}
1 \\
1 \\
0 \\
1 \\
0 \\
0 \\
0 \\
1 \\
0
\end{bmatrix}, \bY^\prime = \begin{bmatrix}
0 \\
0 \\
1 \\
0 \\
1 \\
1 \\
1 \\
0 \\
1
\end{bmatrix}
\end{equation}

Assume that the defective items are 1 and 2, i.e., $\bX = [1, 1, 0, 0, 0, 0, 0, 0, 0]^T$; then the observed vector is $\mathbf{z} = [\bY^T \ \overline{\bY}^T]^T$. Using the three rules in the proof of Lemma~\ref{lem:specialU}, we obtain vector $\bY^\prime$. We note that $\bY^\prime = \cM_1 \bigvee \cM_2 = \cM \odot \bX$. Using a decoding algorithm (which is omitted in this example), we can identify items 1 and 2 as defective items from $\bY^\prime$.

\subsection{Encoding procedure}
\label{sub:enc}

To implement the divide and conquer strategy, we need to divide the set of defective items into small subsets such that defective items in those subsets can be effectively identified. We define $\kappa = \left\lceil \frac{|S|}{u} \right\rceil \geq 1$ as an integer, and create an $h \times n$ matrix $\cG$ containing $\kappa$ rows, denoted as $i_1, i_2, \ldots, i_\kappa$, with probability at least $1 - \epsilon$ such that (i) $|S_{i_1}| = \cdots = |S_{i_\kappa}| = u$ and (ii) $S_{i_1} \cup \ldots \cup S_{i_\kappa} = S$ for any $\epsilon \geq 0$ where $S_i$ is the set of indices of defective items in row $\cG_{i, *}$. For example, if $n = 6$, the defective items are 1, 2, and 3, and $\cG_{1, *} = (1, 0, 1, 0, 1, 1)$, then $S_1 = \{1, 3 \}$. These conditions guarantee that all defective items will be included in the decoded set.

To achieve such a $\cG$, for any $|S| \leq d$, a pruning matrix $\cG^\prime$ of size $h \times d$ after removing $n-d$ columns $\cG_x$ for $x \in [n] \setminus S$ must be an $(u, d-u)$-completely separating matrix with high probability. From Definition~\ref{def:separatingMatrix}, the matrix $\cG^\prime$ is also an $(u, |S| - u)$-completely separating matrix. Then, the $\kappa$ rows are chosen as follows. We choose a collection of sets of defective items: $P_l = \{j_{(l-1)u + 1}, \ldots, j_{lu} \}$ for $l = 1, \ldots, \kappa - 1$. $P^\prime$ is a set satisfying $P^\prime \subseteq \cup_{l=1}^{\kappa - 1}P_l$ and $|P^\prime| = \kappa u - |S|$. Then we pick the last set as follows: $P_\kappa = \left( S \setminus \cup_{l=1}^{\kappa - 1}P_l \right) \cup P^\prime$. From Definition~\ref{def:separatingMatrix}, for any $P_l$, there exists a row, denoted $i_l$, such that $g_{i_l x} = 1$ for $x \in P_l$ and $g_{i_l y} = 0$ for $y \in S \setminus P_l$, where $l = 1, \ldots, \kappa$. Then, we have $S_{i_l} = P_l$ and row $i_l$ is singular to sets $S_{i_l}$ and $S \setminus S_{i_l}$ for $l = 1, \ldots, \kappa$. Condition (i) thus holds. Condition (ii) also holds because $\cup_{l=1}^{\kappa}S_{i_l} = \cup_{l=1}^{\kappa}P_l = S$. The matrix $\cG$ is specified in section~\ref{sec:main}.

After creating the matrix $\cG$, we generate matrix $\mathcal{A}$ as in \eqref{eqn:elementaryMatrix}. Then the final measurement matrix $\cT$ of size $(2k + 1)h \times n$ is created as follows:

\begin{equation}
\label{eqn:meausrementMatrix}
\cT = \begin{bmatrix}
\cG_{1, *} \\
\mathcal{A} \times \diag(\cG_{1, *}) \\
\vdots \\
\cG_{h, *} \\
\mathcal{A} \times \diag(\cG_{h, *})
\end{bmatrix}
= \begin{bmatrix}
\cG_{1, *} \\
\cM \times \diag(\cG_{1, *}) \\
\overline{\cM} \times \diag(\cG_{1, *}) \\
\vdots \\
\cG_{h, *} \\
\cM \times \diag(\cG_{h, *}) \\
\overline{\cM} \times \diag(\cG_{h, *})
\end{bmatrix}
\end{equation}

The vector observed using $u$-NATGT after performing the tests given by the measurement matrix $\cT$ is
\begin{eqnarray}
\bY = \cT \otimes \bX &=& \begin{bmatrix}
\cG_{1, *} \\
\mathcal{A} \times \diag(\cG_{1, *}) \\
\vdots \\
\cG_{h, *} \\
\mathcal{A} \times \diag(\cG_{h, *})
\end{bmatrix} \otimes \bX
= \begin{bmatrix}
\cG_{1, *} \otimes \bX \\
\mathcal{A} \otimes \bX_1 \\
\vdots \\
\cG_{h, *} \otimes \bX \\
\mathcal{A} \otimes \bX_h
\end{bmatrix} \notag \\
&=& \begin{bmatrix}
\cG_{1, *} \otimes \bX\\
\cM \otimes \bX_1  \\
\overline{\cM} \otimes \bX_1 \\
\vdots \\
\cG_{h, *} \otimes \bX\\
\cM \otimes \bX_h \\
\overline{\cM} \otimes \bX_h \\
\end{bmatrix}
= \begin{bmatrix}
y_1 \\
\bY_1 \\
\overline{\bY}_1 \\
\vdots \\
y_h \\
\bY_h \\
\overline{\bY}_h
\end{bmatrix}
= \begin{bmatrix}
y_1 \\
\mathbf{z}_1 \\
\vdots \\
y_h \\
\mathbf{z}_h
\end{bmatrix} \label{eqn:encoding}
\end{eqnarray}
where $\bX_i = \diag(\cG_{i, *}) \times \bX$, $y_i = \cG_{i, *} \otimes \bX$, $\bY_i = \cM \otimes \bX_i \myeq (y_{i1}, \ldots, y_{ik})^T$, $\overline{\bY}_i = \overline{\cM} \otimes \bX_i \myeq (\overline{y}_{i1}, \ldots, \overline{y}_{ik})^T$, and $\mathbf{z}_i = [\bY_i^T \ \overline{\bY}_i^T]^T$ for $i = 1, 2, \ldots, h$.

We note that $\bX_i$ is the vector representing the defective items corresponding to row $\cG_{i, *}$. If $\bX_i = (x_{i1}, x_{i2}, \ldots, x_{in})^T$, then $S_i = \{ l \mid x_{il} = 1, l \in [n] \}$. We thus have $|S_i| = |\bX_i| \leq d$. Moreover, the condition $y_i = 1$ holds if and only if $|\bX_i| \geq u$.

\subsection{The decoding procedure}
\label{sub:dec}

The decoding procedure is summarized as Algorithm~\ref{alg:decodingThreshold}, where $\bY_i^\prime = (y_{i1}^\prime, \ldots, y_{ik}^\prime)^T$ is presumed to be $\cM \odot \bX_i$. The procedure is briefly explained as follows: Step~\ref{alg:scan} enumerates the $h$ rows of $\cG$. Step~\ref{alg:checkPositive} checks if there are at least $u$ defective items in row $\cG_{i, *}$. Steps~\ref{alg:convert2CNAGT} to~\ref{alg:endConverting2CNAGT} calculate $\bY_i^\prime$, and Step~\ref{alg:eliminateFPs} checks if all items in $G_i$ are truly defective and adds them into $S$. 

\begin{algorithm}
\caption{Decoding procedure for $u$-NATGT}
\label{alg:decodingThreshold}
\textbf{Input:} Outcome vector $\bY$, $\cM$.\\
\textbf{Output:} The set of defective items $S$.

\begin{algorithmic}[1]
\STATE $S = \emptyset$. \label{alg:init}
\FOR {$i=1$ \TO $h$} \label{alg:scan}
	\IF {$y_i = 1$} \label{alg:checkPositive}
		\FOR {$l = 1$ \TO $k$} \label{alg:convert2CNAGT}
			\IF {$y_{il} = 1$} \STATE $y^\prime_{il} = 1$ \ENDIF
			\IF {$y_{il} = 0$ and $\overline{y}_{il} = 1$} \STATE $y^\prime_{il} = 0$ \ENDIF
			\IF {$y_{il} = 0$ and $\overline{y}_{il} = 0$} \STATE $y^\prime_{il} = 1$ \ENDIF
		\ENDFOR \label{alg:endConverting2CNAGT}
		\STATE Decode $\bY^\prime_i$ using $\cM$ to get the defective set $G_i$. \label{alg:decodeAll}
		\IF {$|G_i| = u$ and $\bigvee_{j \in G_i} \cM_j \equiv \bY^\prime_i$} \label{alg:eliminateFPs}
			\STATE $S = S \cup G_i$.
		\ENDIF \label{alg:endEliminateFPs}
	\ENDIF	
\ENDFOR
\STATE Return $S$. \label{alg:defectiveSet}
\end{algorithmic}
\end{algorithm}

\subsection{Correctness of the decoding procedure}
\label{sub:DC}
Recall that our objective is to recover $\bX_i$ from $y_i$ and $\mathbf{z}_i = \begin{bmatrix} \bY_i \\ \overline{\bY}_i \end{bmatrix}$ for $i = 1, 2, \ldots, h$. Step~\ref{alg:scan} enumerates the $h$ rows of $\cG$. We have that $y_i$ is the indicator for whether there are at least $u$ defective items in row $\cG_{i, *}$. If $y_i = 0$, it implies that there are less than $u$ defective items in row $\cG_{i, *}$. Since we only focus on row $\cG_{i, *}$ which has exactly $u$ defective items, vector $\mathbf{z}_i$ is not considered if $y_i = 0$. This is achieved by Step~\ref{alg:checkPositive}.

When $y_i = 1$, there are at least $u$ defective items in row $\cG_{i, *}$. If there are exactly $u$ defective items in this row, they are always identified as described by Lemma~\ref{lem:specialU}. Our task is now to prevent false defectives by decoding $\bY^\prime_i$.

Steps~\ref{alg:convert2CNAGT} to~\ref{alg:endConverting2CNAGT} calculate $\bY^\prime_i$ from $\mathbf{z}_i$. If there are exactly $u$ defective items in row $\cG_{i, *}$, we have $\bY_i^\prime = \cM \otimes \bX_i$ and $|\bX_i| = u$. If there are more than $u$ defective items in row $\cG_{i, *}$, vector $\bY_i^\prime = \cM \otimes \bX_i^\prime$ for some vector $\bX_i^\prime \in \{0, 1 \}^{n}$ after implementing Steps~\ref{alg:convert2CNAGT} to~\ref{alg:endConverting2CNAGT}. In the latter case, we may not recover $\bX_i^\prime$ correctly. Moreover, it is not clear whether the non-zero entries in $\bX_i^\prime$ are necessarily the indices of defective items. Therefore, our task is to decode $\bY^\prime_i$ using matrix $\cM$ to get the defective set $G_i$, and then validate whether all items in $G_i$ are defective.

There exists at least $\kappa$ rows of $\cG$ in which there are exactly $u$ defective items, and we need to identify all defective items in these rows. Therefore, we only consider the case when the number of defective items obtained from decoding $\bY^\prime_i$ is equal to $u$; i.e., $|G_i| = u$. Our task is now to avoid false positives, which is accomplished by Step~\ref{alg:eliminateFPs}. There are two sets of defective items corresponding to $\mathbf{z}_i$: the first one is the true set, which is $S_i$ and is \textit{unknown}, and the second one is $G_i$, which is expected to be $S_i$ (albeit not surely) and $|G_i| = u$. Note that $|S_i| \geq u$ because $y_i = 1$. If $G_i \equiv S_i$, we can always identify $u$ defective items and the condition in Step~\ref{alg:eliminateFPs} always holds because of Lemma~\ref{lem:specialU}. We need to consider the case $G_i \not\equiv S_i$; i.e., when there are more than $u$ defective items in row $\cG_{i, *}$. We break down this case into two categories:
\begin{enumerate}
\item When $|G_i \setminus S_i| = 0$: in this case, all elements in $G_i$ are defective items. We do not need to consider whether $\bigvee_{j \in G_i} \cM_j \equiv \bY^\prime_i$. If this condition holds, we obtain the true defective items. If it does not hold, we do not take $G_i$ into the set of defective items.
\item When $|G_i \setminus S_i| \neq 0$: in this case, we prove that $\bigvee_{j \in G_i} \cM_j \equiv \bY^\prime_i$ does not hold; i.e., none of the elements in $G_i$ are added to the defective set. Consider any $j_1 \in G_i \setminus S_i$ and $j_2 \in G_i \setminus \{j_1 \}$. Since $|S_i| \leq d$ and $\cM$ is an $(d+1)$-disjunct matrix, there exists a row, denoted $\tau$, such that $m_{\tau j_1} = 1, m_{\tau j_2} = 0$, and $m_{\tau x} = 0$ for $x \in S_i$. On the other hand, because $|G_i| = u$ and $|S_{i}| \leq d$, there are less than $u$ defective items in row $\tau$; i.e., $y_{i \tau} = 0$. Because $u \leq |S_i|$, we have $\overline{y}_{i \tau} = 1$, which implies that $y^\prime_{i \tau} = 0$. However, we have $\bigvee_{x \in G_i} m_{\tau x} = \left( \bigvee_{x \in G_i \setminus \{j_1 \} } m_{\tau x} \right) \bigvee m_{\tau j_1} = \left( \bigvee_{x \in G_i \setminus \{j_1 \} } m_{\tau x} \right) \bigvee 1 = 1 \neq 0 = y^\prime_{i \tau}$. Therefore, the condition $\bigvee_{j \in G_i} \cM_j \equiv \bY^\prime_i$ does not hold.
\end{enumerate}
Thus, Step~\ref{alg:eliminateFPs} eliminates false positives. Finally, Step~\ref{alg:defectiveSet} returns the defective item set $S$.

\subsection{The decoding complexity}
\label{sub:decodingComplex}
Because $\cT$ is constructed using $\cG$ and $\cM$, the probability of successful decoding of $\bY$ depends on these choices. Given an input vector $\bY^\prime_i$, we get the set of defective items from decoding of $\cM$. The probability of successful decoding of $\bY$ thus depends only on $\cG$. Since $\cG$ has $\kappa$ rows satisfying (i) and (ii) with probability at least $1 - \epsilon$, all $|S|$ defective items can be identified by using $t = h(2k + 1)$ tests with probability of at least $1 - \epsilon$ for any $\epsilon \geq 0$.

The time to run Steps~\ref{alg:convert2CNAGT} to~\ref{alg:endConverting2CNAGT} is $O(k)$. It takes $\poly(k)$ to run Step~\ref{alg:decodeAll} as in Theorem~\ref{thr:dDisjunct}. Because each column of $\cM$ can be generated in time $O(k^2 \log{N})$, it takes $u \times O(k^2 \log{N})$ time to run Steps~\ref{alg:eliminateFPs} to~\ref{alg:endEliminateFPs}. Because it runs $h$ times for the loop in Step~\ref{alg:scan}, the total decoding complexity is:
\begin{equation}
h \times \left( O(k) + \poly(k) + u \times O(k^2 \log{N}) \right) =  h \times \poly(k). \nonumber
\end{equation}

We summarize the divide and conquer strategy in the following theorem:

\begin{theorem}
\label{thr:general}
Let $2 \leq u \leq d < n$ be integers and $S$ be the defective set. Suppose that an $h \times n$ matrix $\cG$ contains $\kappa$ rows, denoted as $i_1, \ldots, i_\kappa$, such that (i) $|S_{i_1}| = \cdots = |S_{i_\kappa}| = u$ and (ii) $S_{i_1} \cup \ldots \cup S_{i_\kappa} = S$, where $S_{i_l}$ is the index set of defective items in row $\cG_{i_l, *}$. Suppose that an $k \times n$ matrix $\cM$ is an $(d + 1)$-disjunct matrix that can be decoded in time $O(\mathsf{A})$ and each column of $\cM$ can be generated in time $O(\mathsf{B})$. Then an $(2k + 1)h \times n$ measurement matrix $\cT$, as defined in \eqref{eqn:meausrementMatrix}, can be used to identify up to $d$ defective items in $u$-NATGT in time $O(h \times (\mathsf{A} + u \mathsf{B}))$.

The probability of successful decoding depends only on the event that $\cG$ has $\kappa$ rows satisfying (i) and (ii). Specifically, if that event happens with probability at least $1 - \epsilon$, the probability of successful decoding is also at least $1 - \epsilon$ for any $\epsilon \geq 0$.
\end{theorem}

\section{Complexity of proposed scheme}
\label{sec:main}
We specify the matrix $\cG$ in Theorem~\ref{thr:general} to get the desired number of tests and decoding complexity for identifying up to $d$ defective items. Note that when $u = d$, the number of defective items should be $u$ (otherwise, every test would yield a negative outcome). In this case, Lemma~\ref{lem:specialU} is sufficient to find the defective items. We consider the following notions of deterministic and randomized decoding. Deterministic decoding is a scheme in which all defective items are found with probability 1. It is achievable when \textbf{every} $h \times d$ submatrix of $\cG$ is $(u, d- u)$-completely separating. Randomized decoding reduces the number of tests, in which all defective items can be found with probability at least $1 - \epsilon$ for any $\epsilon > 0$. It is achieved when \textbf{each} $h \times d$ submatrix is an $(u, d- u)$-completely separating matrix with probability at least $1 - \epsilon$.

\subsection{Deterministic decoding}
\label{sub:detDec}
The following theorem states that there exists a deterministic algorithm for identifying all defective items by choosing $\cG$ of size $h \times n$ to be an $(u, d-u)$-completely separating matrix in Theorem~\ref{thr:existingSeparatingMatrix}.
\begin{theorem}
\label{thr:deterministicMain}
Let $2 \leq u \leq d \leq n$. There exists a $t \times n$ matrix such that up to $d$ defective items in $u$-NATGT can be identified in time $t \times \poly(d^2 \log{n})$, where
\begin{equation}
t = O \left( \left( \frac{d}{u} \right)^u \left( \frac{d}{d -u} \right)^{d - u} \cdot  d^3 \log{n} \cdot \log{\frac{n}{d}} \right) \nonumber
\end{equation}
\end{theorem}

\begin{proof}
On the basis of Theorem~\ref{thr:general}, a $t \times n$ measurement matrix $\cT$ is generated as follows:
\begin{enumerate}
\item Choose an $h \times n$ $(u, d-u)$-completely separating matrix $\cG$ as in Theorem~\ref{thr:existingSeparatingMatrix}, where

$h = \left( \frac{d}{u} \right)^u \left( \frac{d}{d -u} \right)^{d - u} \left( d \log{\frac{\mathrm{e}n}{d}} + u\log{\frac{\mathrm{e}d}{u}} \right) + 1$.
\item Choose a $k \times n$ $(d + 1)$-disjunct matrix $\cM$ as in Theorem~\ref{thr:dDisjunct}, where $k = O(d^2 \log{n})$ and the decoding time of $\cM$ is $\poly(k)$.
\item Define $\cT$ as~\eqref{eqn:meausrementMatrix}.
\end{enumerate}

Since $\cG$ is an $h \times d$ $(u, d-u)$-completely separating matrix, for any $|S| \leq d$, an $h \times d$ pruning matrix $\cG^\prime$, which is created by removing $n-d$ columns $\cG_x$ for $x \in [n] \setminus S$, is also an $(u, d-u)$-completely separating matrix with probability 1. From Definition~\ref{def:separatingMatrix}, matrix $\cG^\prime$ is also an $(u, |S| - u)$-completely separating matrix. Then, there exists $\kappa$ rows satisfying (i) and (ii) as described in section~\ref{sub:enc}. From Theorem~\ref{thr:general}, up to $d$ defective items can be recovered using $t = h \cdot O(d^2 \log{n})$ tests with probability 1, in time $h \cdot \poly(k)$.
\end{proof}

\subsection{Randomized decoding}
\label{sub:rndDec}
For randomized decoding, matrix $\cG$ is chosen such that the pruning matrix $\cG^\prime$ of size $h \times d$ created by removing $n - d$ columns $\cG_x$ of $\cG$ for $x \in [n] \setminus S$ is an $(u, d-u)$-completely separating matrix with probability at least $1 - \epsilon$ for any $\epsilon > 0$. This results is an improved number of tests and decoding time compared to Theorem~\ref{thr:deterministicMain}:

\begin{theorem}
\label{thr:randomizedMain}
Let $2 \leq u \leq d \leq n$. For any $\epsilon > 0$, up to $d$ defective items in $u$-NATGT can be identified using
\begin{eqnarray}
t &=& O \left( \left( \frac{d}{u} \right)^u \left( \frac{d}{d - u} \right)^{d-u} \left(u \log{\frac{d}{u}} + \log{\frac{1}{\epsilon}} \right) \cdot d^2 \log{n} \right) \notag
\end{eqnarray}
tests with probability at least $1 - \epsilon$. The decoding time is $t \times \poly(d^2 \log{n})$.
\end{theorem}

\begin{proof}
Using Theorem~\ref{thr:general}, a $t \times n$ measurement matrix $\cT$ is generated as follows:
\begin{enumerate}
\item Choose an $h \times n$ matrix $\cG$ as in Corollary~\ref{cor:rndMatrix}, where $h = \left( \frac{d}{u} \right)^u \left( \frac{d}{d - u} \right)^{d-u} \left( u \log{\frac{ \mathrm{e} d}{u}} + \log{\frac{1}{\epsilon}} \right) $.
\item Generate a $k \times n$ $(d + 1)$-disjunct matrix $\cM$ using Theorem~\ref{thr:dDisjunct}, where $k = O(d^2 \log{n})$ and the decoding time of $\cM$ is $\poly(k)$.
\item Define $\cT$ as~\eqref{eqn:meausrementMatrix}.
\end{enumerate}

Let $\cG$ be an $h \times n$ matrix as described in Corollary~\ref{cor:rndMatrix}. Then for any $|S| \leq d$, an $h \times d$ pruning matrix $\cG^\prime$, which is created by removing $n-d$ columns $\cG_x$ for $x \in [n] \setminus S$, is an $(u, d-u)$-completely separating matrix with probability at least $1 - \epsilon$. From Definition~\ref{def:separatingMatrix}, matrix $\cG^\prime$ is also an $(u, |S| - u)$-completely separating matrix. Then, there exist $\kappa$ rows satisfying (i) and (ii) as described in section~\ref{sub:enc} with probability at least $1 - \epsilon$. From Theorem~\ref{thr:general}, all $|S|$ defective items can be recovered using $t = h \cdot O(d^2 \log{n})$ tests with probability at least $1-\epsilon$ and in time $h \cdot  \poly(k)$.
\end{proof}

\section{Simulation}
\label{sec:simul}

In this section, we visualize the decoding times in Table~\ref{tbl:comparison}. For deterministic decoding, the number of items $n$ and the maximum number of defective items $d$ are set to be $\{2^{20}, 2^{30}, 2^{40}, 2^{50}, 2^{60}\}$ and $\{100, 1,000 \}$, respectively. For the randomized algorithm, the number of items $n$ and the maximum number of defective items $d$ are set to be\footnote{%
We note that the parameters are chosen for theoretical benchmarks and do not necessarily reflect the range encountered for practical applications.
} $\{2^{30}, 2^{50}, 2^{100}, 2^{300}, 2^{500}\}$ and $\{10, 100, 1,000 \}$, respectively. The threshold $u$ is set to be $0.2 d$. Finally, the precision $\epsilon$ for randomized algorithms is set to be $10^{-6}$.

\begin{figure}
\centering
\begin{minipage}{.5\textwidth}
  \centering
  \includegraphics[width=1.00\linewidth]{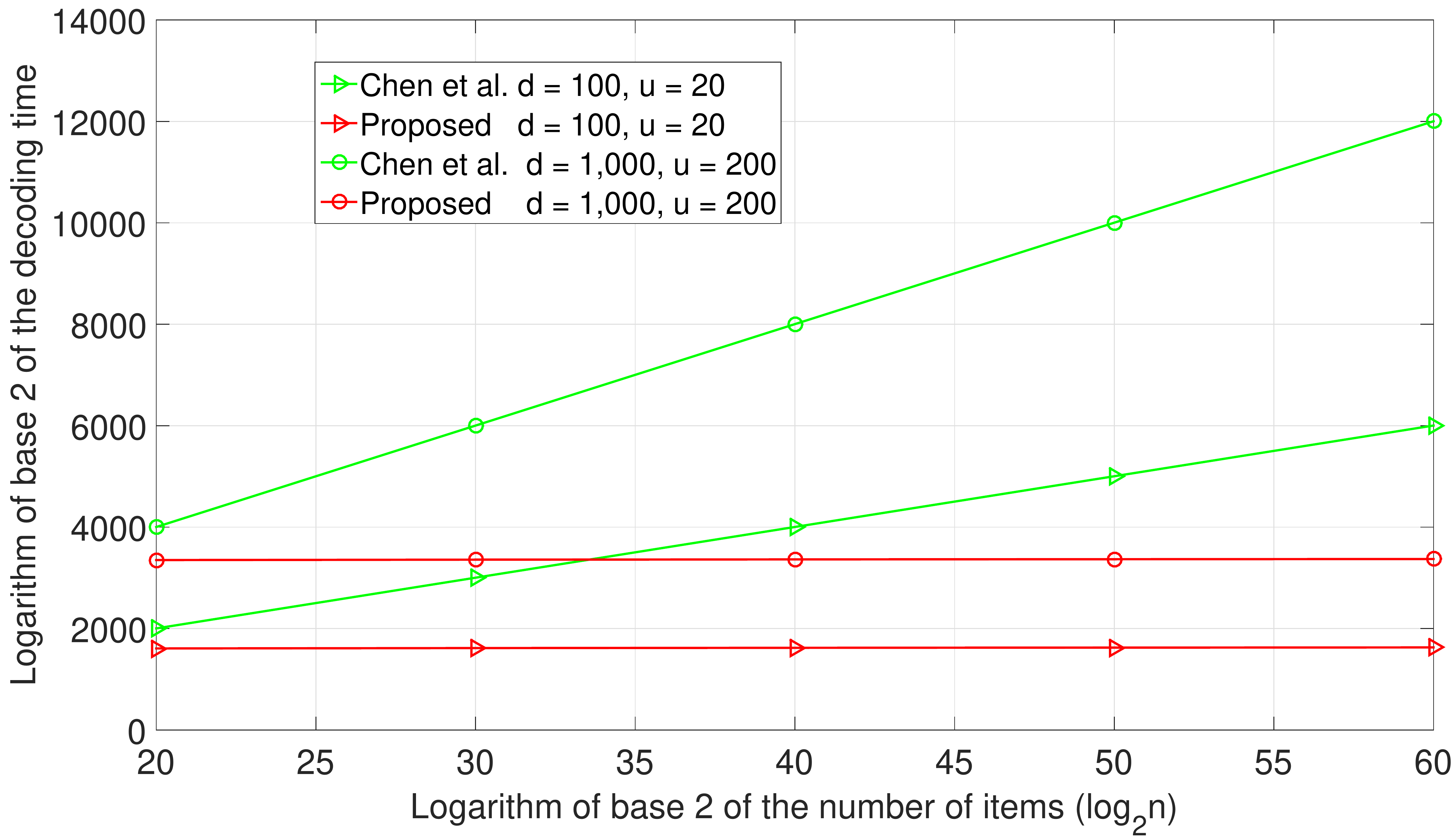}
  \captionof{figure}{Decoding time in deterministic decoding}
  \label{fig:det}
\end{minipage}%
\begin{minipage}{.5\textwidth}
  \centering
  \includegraphics[width=1.00\linewidth]{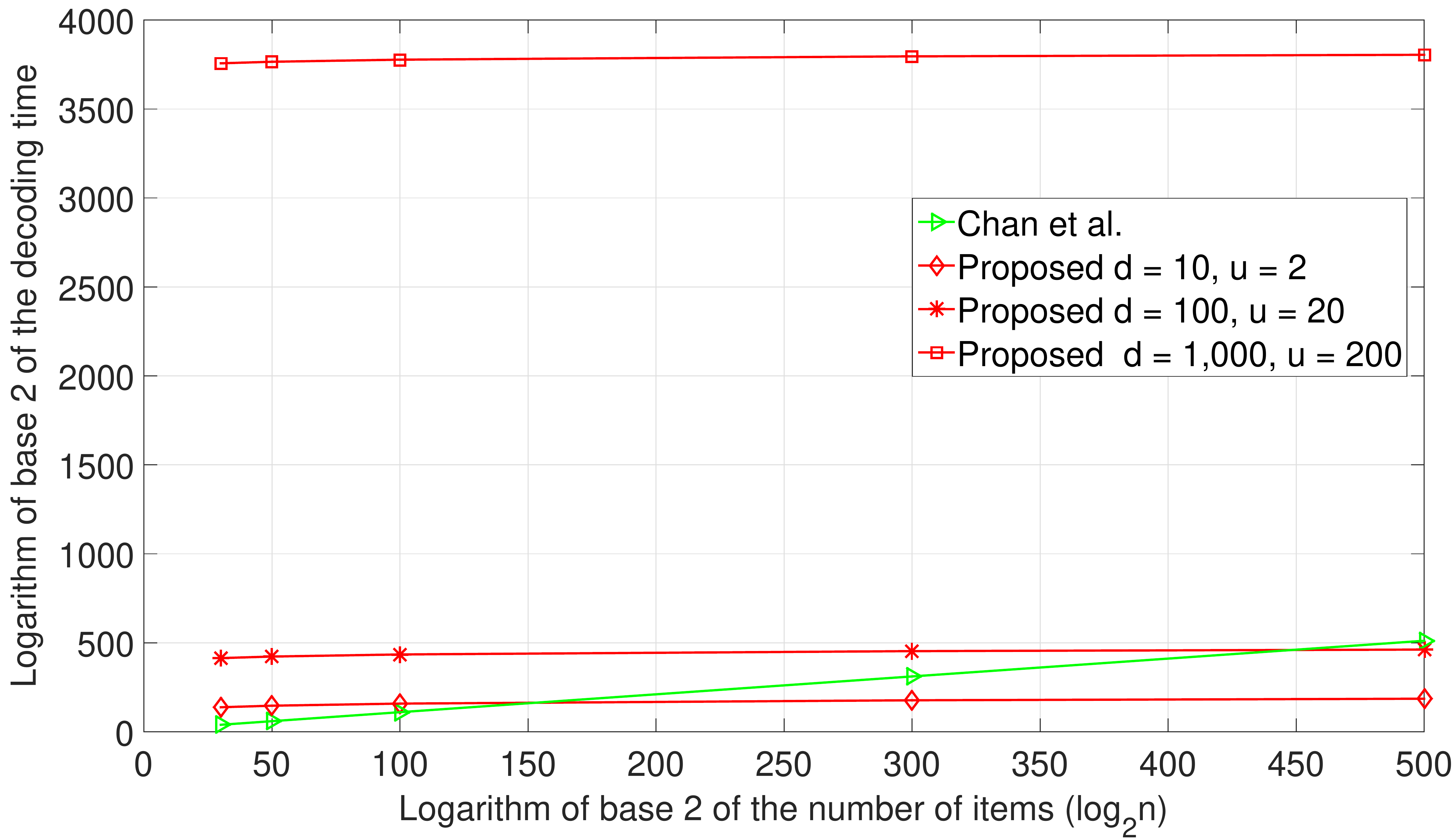}
  \captionof{figure}{Decoding time in randomized decoding}
  \label{fig:rnd}
\end{minipage}
\end{figure}

We see that for deterministic decoding, our proposed scheme is always better than the one proposed by Chen et al.~\cite{chen2009nonadaptive} as shown in Fig.~\ref{fig:det}. However, for randomized decoding, our proposed scheme is better than the one proposed by Chan et al.~\cite{chan2013stochastic} for $d \leq \log{n}$ and large enough $n$, as shown in Fig.~\ref{fig:rnd}. Since the decoding time in~\cite{chan2013stochastic} is not affected much by the parameters $d$ and $u$, we only plot one graph for it. Note that when $n \leq 2^{60}$, the decoding time in our proposed scheme is worse than the one in~\cite{chan2013stochastic}. The main reason is that the decoding time of a $d$-disjunct matrix in Theorem~\ref{thr:dDisjunct} is high, i.e., $O(d^{11} \log^{17}{n})$. Therefore, if there exists any $d$-disjunct matrix with low decoding complexity, e.g., $O(d^2 \log^2{n})$, our proposed scheme would be much better than the one in~\cite{chan2013stochastic} when the number of items $n$ is small.

\section{Conclusion}
\label{sec:cls}

We introduced an efficient scheme for identifying defective items in NATGT. Its main idea is to convert the test outcomes in NATGT to CNAGT by distributing defective items into tests properly. Then all defective items are identified by using some known decoding procedure in CNAGT. However, the algorithm works only for $g = 0$. Extending the results to $g > 0$ is left for future work. Since the number of tests in the randomized decoding is quite large, reducing it is also an important task. Moreover, it would be interesting to consider noisy NATGT as well, in which erroneous tests are present in the test outcomes.

\section{Acknowledgement}
The first author thanks SOKENDAI for supporting him via The Short-Stay Abroad Program 2017.

\bibliographystyle{ieeetr}
\bibliography{bibli}

\end{document}